\documentclass[letterpaper, 10 pt, conference]{ieeeconf} 

\IEEEoverridecommandlockouts                              

\usepackage{amsmath} 
\usepackage{amssymb}  
\usepackage{tikz}
\usepackage[linesnumbered,ruled,vlined]{algorithm2e}
\usepackage{soul}
\usepackage{subcaption} 
\usepackage[hidelinks]{hyperref}
\usepackage{cleveref}

\usepackage[noadjust]{cite}

\usetikzlibrary{shapes,arrows,positioning}
\usepackage{color}

\usepackage{amsthm}

\makeatletter
\newcommand{\removelatexerror}{\let\@latex@error\@gobble}
\makeatother

\newtheorem{remark}{Remark}[]

\DeclareMathOperator*{\argmin}{arg\,min}

\title{\LARGE \bf
A Novel Warehouse Multi-Robot Automation System with Semi-Complete and Computationally Efficient Path Planning and Adaptive Genetic Task Allocation Algorithms
}

\author{Kam Fai Elvis Tsang, Yuqing Ni, Cheuk Fung Raphael Wong and Ling Shi
\thanks{Kam Fai Elvis Tsang, Yuqing Ni, Cheuk Fung Raphael Wong and Ling Shi are with the Department of Electronic and Computer Engineering, Hong Kong University of Science and Technology, Hong Kong. Emails: {\tt\small \{kftsang, yniac, cfrwong, eesling\}@ust.hk}}%
\thanks{The work is supported by a Hong Kong ITC research fund ITS/066/17FP-A.}
}

\begin{document}

\newtheorem{lemma}{Lemma}
\newtheorem{theorem}{Theorem}
\newtheorem{problem}{Problem}
\newtheorem{example}{Example}
\newtheorem{assumption}{Assumption}
\newtheorem{definition}{Definition}

\maketitle
\thispagestyle{empty}
\pagestyle{empty}

\setlength{\arraycolsep}{0.15em}

\begin{abstract}

We consider the problem of warehouse multi-robot automation system in discrete-time and discrete-space configuration with focus on the task allocation and conflict-free path planning. We present a system design where a centralized server handles the task allocation and each robot performs local path planning distributively. A genetic-based task allocation algorithm is firstly presented, with modification to enable heuristic learning. A semi-complete potential field based local path planning algorithm is then proposed, named the recursive excitation/relaxation artificial potential field (RERAPF). A mathematical proof is also presented to show the semi-completeness of the RERAPF algorithm. The main contribution of this paper is the modification of conventional artificial potential field (APF) to be semi-complete while computationally efficient, resolving the traditional issue of incompleteness. Simulation results are also presented for performance evaluation of the proposed path planning algorithm and the overall system.
\end{abstract}

\section{Introduction}
In the recent decades, the rapid advancement of multi-robot systems has been attractive to both academia and industries because of the wide range of potential applications. Some of the promising applications include transportation, industrial plant inspection \cite{mrta_mrpp_astar_ga} and logistic management \cite{mrta_hospital}. In particular, conflict-free path planning algorithm \cite{mrpp_mstar} and task allocation policy \cite{mrta_hospital} are the main challenges in multi-robot automation system such as warehouse management which will be the focus of this paper. Unfortunately, both of these problems are in general NP-hard \cite{mrta_mrpp}, thus of great incentives to investigate.

Multi-robot task allocation refers to the process of allocating tasks and the execution orders for each robot in the system. Although it can normally be considered as a combinatorial optimization problem \cite{mrta_cnp_nn}, it is difficult to solve efficiently due to the large search space. Yuan et al. \cite{mrta_cnp_nn} introduced a neural network based solution for effective bidding in auction at the expense that a large training set is necessary. Liu and Kroll \cite{mrta_mrpp_astar_ga} implemented genetic algorithm (GA) to generate task allocations but the performance greatly depends on the quality of heuristics for the fitness function. Whilst each of these approaches have their merits, their respective flaws have limited their performances to a large extent.

Path planning algorithms can be categorized into global and local planning \cite{rpp_svm}. Global planning assumes all information are available and plan the robot path accordingly. Despite the promising solutions contributed by global information, it is computationally expensive thus typically infeasible for large scale system to operate at real time. Local planning, on the other hand, mainly considers the immediate or recent ambience of the robot for online planning with low computation load. A widely used computationally efficient path planning algorithm is the artificial potential field (APF) algorithm \cite{mrpp_apf}, albeit being incomplete due to local minima which has been a challenging problem \cite{rpp_apf_deadlock}, the details of which will be discussed in depth in \autoref{section:main}. Numerous attempts have been made to overcome this issue. For example, Tuazon et al. \cite{mrpp_apf_fuzzy} proposed the integration of fuzzy logic into the APF algorithm to identify and escape from local minimum. Kov\'{a}cs et al. \cite{mrpp_apf_animal} implemented the BUG algorithm in additional to APF method as a complement. While these approaches showed effectiveness in simulation, seldom did they prove the completeness of the algorithms. 

This paper presents a semi-complete and computationally efficient potential-based local path planning algorithm, named the recursive excitation/relaxation artificial potential field algorithm, hereinafter referred to as the RERAPF algorithm. It has completely eradicated the incompleteness flaw of APF algorithm, and the proof of completeness will be presented in later sections. In addition, we also present a genetic-based task allocation algorithm and an adaptive integrated system with learning ability for the fitness function of the genetic algorithm to improve the overall system performance. We show in simulation the significant improvement in computation speed of RERAPF compared with A* algorithm, and the overall performance of the integrated system.

The remaining of the paper is organized as follows. Some preliminaries and notations are defined in \autoref{section:prelim}, then the problem setup is introduced in \autoref{section:problem}. The main results, including the proposed algorithms and proof of completeness, are presented in \autoref{section:main} and the simulation results in \autoref{section:sim}.

\section{Preliminaries and Notations}
\label{section:prelim}

In this paper, a warehouse is defined as a grid world $\mathcal{W}$, an example layout of which is shown in \autoref{fig:layout}, with $N$ homogenous autonomous robots, $M$ static obstacles (including storage shelves and walls), $P$ reachable positions and $K$ tasks. Denote $\mathcal{R} = \{r_1, \dots, r_N \}$ as the set of robots, $\mathcal{O} = \{o_1, \dots, o_M\} $ the set of obstacles, $\mathcal{T} = \{t_1, \dots, t_K \}$ the set of tasks and $\mathcal{S} = \{s_1, \dots s_P\} \subset \mathbb{Z}^2$ the set of all reachable positions. Note that $\mathcal{W} = \mathcal{S} \cup \mathcal{O}$. A position $s$ is represented by a positional vector $[x,y]' \in \mathcal{S}$. Each robot $r_i \in \mathcal{R}$ has a position $s_{r_i}(k)$ at time $k$, a goal position $s_g^i$ and an ordered set of $K_i$ tasks $T_i = \{ t_{i,1}, \dots t_{i,K_i} \} \subset \mathcal{T}$. Also each obstacle $o_i \in \mathcal{O}$ and task $t_i \in \mathcal{T}$ has a position $s_{o_i}$ and $s_{t_i}$ respectively. When $T_i \neq \emptyset$, the goal position for $r_i$ is simply the position of first task, i.e., $s_g^i = s_{t_{i,1}}$. 

\begin{figure}[tbp]
\def\gridsize{0.3}
\def\gridwidth{19}
\def\gridheight{21}
\centering
\begin{tikzpicture}[box/.style={rectangle,draw=black,very thick, minimum size=\gridsize cm}, robot/.style={circle,fill=green,minimum size=0.25 cm,inner sep=0pt}]
\foreach \i   [evaluate=\i as \x using \i*\gridsize] in {0,...,\gridwidth} {
	\foreach \j   [evaluate=\j as \y using \j*\gridsize] in {0,...,\gridheight} {
        \node[box] at (\x,\y){};
    }
}

\foreach \j   [evaluate=\j as \y using \j*\gridsize] in {0,...,\gridheight} {
    \node[box, fill=red] at (0,\y){};
    \node[box, fill=red] at (5.7,\y){};
}

\foreach \i   [evaluate=\i as \x using \i*\gridsize] in {0,...,\gridwidth} {
    \node[box, fill=red] at (\x,0){};
    \node[box, fill=red] at (\x,6.3){};
}

\foreach \i   [evaluate=\i as \x using \i*\gridsize] in {3,4,7,8,11,12,15,16} {
	\foreach \j   [evaluate=\j as \y using \j*\gridsize] in {3,4,5,6,9,10,11,12,15,16,17,18} {
        \node[box, fill=red] at (\x,\y){};
    }
}

\node[robot] at (3,5.1){};
\node[robot] at (0.3,2.7){};
\node[robot] at (5.4,0.6){};
\node at (3,5.1) {\footnotesize 0};
\node at (0.3,2.7) {\footnotesize 1};
\node at (5.4,0.6) {\footnotesize 2};
\node at (2.4,1.5) {\footnotesize \textcolor{white}{0}};
\node at (1.2,4.8) {\footnotesize \textcolor{white}{1}};
\node at (4.8,2.7) {\footnotesize \textcolor{white}{2}};
\node at (3.6,2.1) {\footnotesize \textcolor{white}{3}};

\node[box,fill=red] at (6.25,5){};  
\node[box,fill=red] at (6.25,4){};  
\node[box] at (6.25,4.5){};
\node[robot] at (6.25, 4.5){};
\node at (6.25,4.5) {\footnotesize 0};
\node at (6.25,4.0) {\footnotesize \textcolor{white}{0}};
\node[anchor=west] at (6.55,5){Obstacles};
\node[anchor=west] at (6.55,4.5){Robots};
\node[anchor=west] at (6.55,4){Tasks};

\end{tikzpicture}
\caption{Warehouse layout}
\label{fig:layout}
\end{figure}
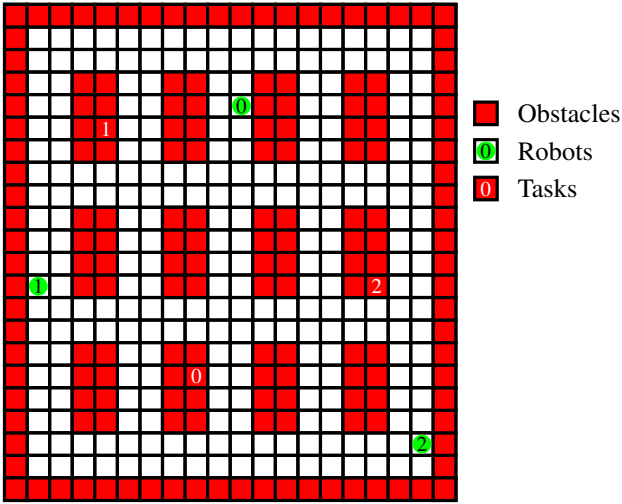

Consider the metric space $M_p=(\mathcal{S}, d_p)$ with metric $d_p(s_i, s_j) = \vert \vert s_i - s_j \vert \vert_p$, i.e., the $p$-norm. We define the neighborhood of $s$, denoted by $\mathcal{N}(s)$, as the closed unit ball in $M_1$ centered at $s$, i.e., $\mathcal{N}(s) = \overline{B_1(s)} = \{s' \in \mathcal{S} \colon d_1(s, s') \leq 1 \}$, and the adjacent neighborhood of $s$ as $\mathcal{N}_a(s) = \mathcal{N}(s)\backslash\{s\}$. Also, the adjacent neighborhood of a space $U$ is defined as $\mathcal{N}_a(U) = \{s' \in \mathcal{S} \colon \min_{s \in U} d_1(s, s') = 1 \}$. The set of successors for the robot $r_i$ at time $k$ is $\mathcal{N}(s_{r_i}(k))$.

\section{Problem Setup}
\label{section:problem}

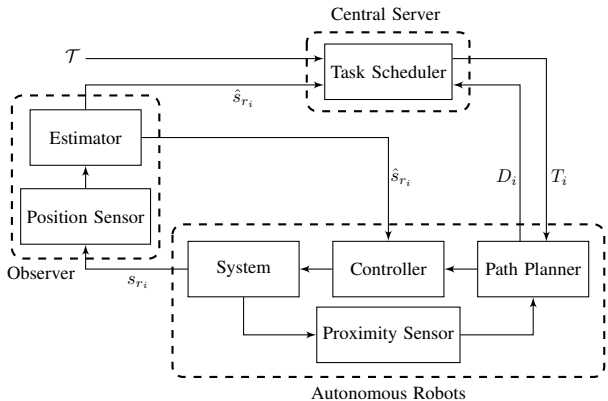
\begin{figure}[h]
	\centering
	\tikzstyle{block} = [draw, fill=white!20, rectangle, minimum height=3em, minimum width=6em]
	\tikzstyle{sum} = [draw, fill=white!20, circle, node distance=1cm]
	\tikzstyle{input} = [coordinate]
	\tikzstyle{output} = [coordinate]
	\tikzstyle{pinstyle} = [pin edge={to-,thin,black}]
	
	\begin{tikzpicture}[auto, node distance=2cm,>=latex', scale=0.7, every node/.style={scale=0.7}]		
		\node [input, name=input1] at (1.5, -0.5) {};
		\node [input, name=input2, below of=input1, node distance=0.5cm] {};
		
		\node [block, right of=input1, node distance=5.75cm, yshift=-0.25cm] (task_scheduler) {Task Scheduler};
		\node [right of=input1, node distance=4.65cm] (task1) {};
		\node [right of=input2, node distance=4.65cm] (task2) {};
		\node [right of=input1, node distance=6.85cm] (task3) {};
		\node [right of=input2, node distance=6.85cm] (task4) {};
		\node [block, below of=input2, node distance=1cm] (estimator) {Estimator};
		\node [block, below of=estimator, node distance=1.5cm] (sensor) {Position Sensor};
		
		\node [block, below of=task_scheduler, node distance=3.75cm] (controller) {Controller};
		\node [block, left of=controller, node distance=2.75cm] (system) {System};
		\node [block, right of=controller, node distance=2.75cm] (path) {Path Planner};
		\node [right of=controller, node distance=2.5cm, yshift=0.395cm] (path1) {};
		\node [right of=controller, node distance=3cm, yshift=0.395cm] (path2) {};
				
		\node [block, below of=controller, node distance=1.25cm] (sensor2) {Proximity Sensor};
	
		\draw [->]	(input1) -- (task1);
		\draw [->]	(estimator) |- (task2);
		\draw [->] 	(sensor) -- (estimator);
		\draw [->] 	(estimator) -| (controller);
		\draw [->] 	(system) -| (sensor);
		\draw [->]	(task3) -| (path2);
		\draw [->]	(path1) |- (task4);
		\draw [->] 	(path) -- (controller);
		\draw [->] 	(controller) -- (system);
		\draw [->]	(system) |- (sensor2);
		\draw [->]	(sensor2) -| (path);
		
		\draw [dashed, thick, rounded corners] (5.7,0) rectangle (8.79,-1.5);
		\draw [dashed, thick, rounded corners] (0.1,-1.2) rectangle (2.9,-4.3);
		\draw [dashed, thick, rounded corners] (3.15,-3.65) rectangle (11.35,-6.55);
		
		\node [align=center] at (2.5,-4.75) {$s_{r_i}$};
		\node [align=center] at (7.5,-2.75) {$\hat{s}_{r_i}$};
		\node [align=center] at (4.5,-1.25) {$\hat{s}_{r_i}$};
		
		\node [align=center] at (7.2, 0.35) {Central Server};
		\node [align=center] at (7.25, -6.85) {Autonomous Robots};
		\node [align=left] at (0.65, -4.575) {Observer};
		
		\node [align=left] at (1.25, -0.45) {$\mathcal{T}$};
		\node [align=left] at (10.5, -2.75) {$T_i$};
		\node [align=left] at (9.5, -2.75) {$D_i$};
		
	\end{tikzpicture}
	\caption{Architecture of proposed warehouse automation system}
	\label{fig:top_level_sys}
\end{figure}
We consider a warehouse multi-robot automation system shown in \autoref{fig:top_level_sys} on a discrete time horizon $k=0,1,\ldots \in \mathbb{N}_0$. The overall system is divided into three subsystems, namely the centralized server, observer and autonomous robots. 

The central server handles incoming tasks, and perform task allocation algorithm to allocate tasks $T_i$ from $\mathcal{T}$ to each robot $r_i$ in a centralized manner. The central server communicates with the robots via a wireless channel, which is assumed to have no bit error or packet drop for simplicity.

Each robot is equipped with proximity sensors to retrieve local information at each time $k$, with a sensing range $\mathcal{P}(s_{r_i}(k)) \subset \mathcal{S}$. We further extend the notation to $\mathcal{P}(s)$ representing the sensing range for any robot at $s$. When $r_i$ receives the task assignment from the server, they will distributively perform path planning to reach $s_g^i$ based on the ambient information retrieved by proximity sensors. When the robot $r_i$ has reached its goal, it will feedback the distance travelled $D_i$ to the central server to improve the task allocation performance. 

The observer consists of position sensor and estimator for estimation of robot positions $s_{r_i}$, denoted by $\hat{s}_{r_i}$. In this paper, we assume a perfect estimator such that $\hat{s}_{r_i} = s_{r_i}$. This subsystem can be an on-board system of the robot such as odometer, or a separate subsystem such as visual camera. 

We aim to design the task scheduler and path planner, which are handled in centralized and distributive manners respectively. To evaluate the system performance, we first define a path cost $J_1$ for the proposed path planning algorithm, to be the ratio of the total travel distance $D_i$ to optimal distance $D_i^*$ (determined by A* search) of each $r_i$. This quantifies the optimality of the path planning algorithm.
\begin{equation} 
	J_1	= \dfrac{\sum_{i=1}^N D_{i}}{\sum_{i=1}^N D_{i}^*}
\end{equation} 
As for the task allocation policy, the paramount concern is the travel distance for each robot $r_i$ with tasks $T_i$. We define an average distance cost function $J_2$ as the average travel distance per robot per task. This is directly correlated to the energy consumption. In addition, the bottleneck for completion time is the longest distance travelled amongst the robots since the system has to wait for the slowest robot to finish its tasks. We further define a $J_3$ as the bottleneck distance per task, as follows:
\begin{eqnarray}
	J_2	&=& \dfrac{1}{KN}\sum_{i=1}^N D_i 	\\
	J_3	&=& \dfrac{1}{K}\max_{i \in [1,N]} D_i
\end{eqnarray}
Finally, we also hope to evaluate the entire system from the perspective of task completion efficiency. Define $J_4$ as the task completion rate, given in
\begin{equation}
	J_4	=	\dfrac{K}{k_{total}}
\end{equation}
where $k_{total}$ is the total time needed to finish all $K$ tasks.
The design problem is to individually minimize $J_1, J_2, J_3$ and maximize $J_4$ to ensure an effective and efficient warehouse automation system.

\section{Main Results}
\label{section:main}
We present our main results, i.e., the proposed genetic multi-robot task allocation algorithm and RERAPF algorithm in this section. The generic GA and conventional APF algorithm are first described in each subsection. The proposed solutions are then built upon each of these algorithms.
\subsection{Genetic Multi-Robot Task Allocation Algorithm}
In order to allocate the tasks effectively and efficiently, a meta-heuristic approach is adopted. In this subsection, genetic algorithm is used along with a proposed learning rule to adapt for the warehouse and multi-robot setup. It should be noted that other meta-heuristic algorithm could also be used with the same learning rule.
\subsubsection{Generic Genetic Algorithm}
In the generic GA, a chromosome is used to represent the potential solution, which, in this case, is the task allocation policy. A fitness function $F$ is defined to evaluate the quality of the chromosomes where higher fitness means better quality. An outline of the standard and generic GA is presented in Algorithm \ref{algo:ga} \cite{Tomioka2007}. A more detailed description of the GA can be found in \cite{Tsang2018}. The detailed design of chromosome representation, fitness function, crossover and mutation operations are described in the following subsections.

\begin{figure}[h]
 \removelatexerror
  \begin{algorithm}[H]
  	\label{algo:ga}
	\caption{Generic Genetic Algorithm}
	$n \leftarrow 0$ \\
	initialize population $P(n)$ \\
	evaluate the fitnesses of $P(n)$ \\
	\While {$n < n_{max}$} {
		choose the best candidates of $P(n)$ to form $P_p(n)$ \\
		crossover $P_p(n)$ to reproduce $P_c(n)$	 \\
		mutate $P_c(n)$ probabilistically	\\
		evaluate the fitnesses of $P_c(n)$ \\
		choose the best candidates as next generation with same population size	\\
		$n \leftarrow n+1$ \\
	}
  \end{algorithm}
\end{figure}

\subsubsection{Chromosome Representation}
The chromosome representation is used to encode tasks allocation policy for the robots. For a set of $N$ robots and $K$ tasks, the length of chromosome is $N+K-1$ and each element of the chromosome is called a genome. In each chromosome, $K$ genomes represents the tasks index, and $N-1$ genomes represent the delimiters of chromosome which are arbitrary negative integers, hereinafter represented by $-1$ to $-(N-1)$. The set of all possible chromosomes $\mathcal{C}$ is all permutations of $\{ 1, 2, \dots, K, -1, -2, \dots, -(N-1) \}$ and each set of tasks are encoded as consecutive genomes, separated by a delimiter. In general, any $C \in \mathcal{C}$ has the form $C = \{c_1, \dots c_{N+K-1} \} = \{T_1, n_1, T_2, n_2, \dots, n_{N-1}, T_{N}\}$, where $n_i$ is a negative integer as delimiter, represents the task allocation policy for robot $r_i$ to $r_N$. For example, the sets of tasks $T_1 = \{t_3, t_5, t_1\}$, $T_2 = \{t_4, t_6\}$, $T_3 = \{t_2, t_7\}$, $T_4 = \emptyset$ can be encoded as $\{3,5,1,-1,4,6,-2,2,7,-3\}$.
\subsubsection{Fitness Function} The fitness function, based on the cost function in \autoref{section:problem}, directly reflects the quality of the chromosome, hence task allocation. A higher value of fitness function means higher quality. The fitness function $F \colon \mathcal{C} \rightarrow \mathbb{R}$ is shown in \eqref{eq:fitness}.
\begin{equation}
	\label{eq:fitness}
	F(C)	=	\bigg( \dfrac{1}{KN}\sum_{i=1}^N \hat{D}_{i}(T_i) + \displaystyle \dfrac{1}{K} \max_i \hat{D}_{i}(T_i) \bigg)^{-1}
\end{equation}
where $\hat{D}_{i}(T_i)$ is the heuristic of the travelled distance for robot $r_i$ with $T_i$, and is given by 
\begin{equation}
	\hat{D}_{i}(T_i) = \hat{d}\big(s_{r_i}(k_0), s_{t_{i,1}}\big) + \sum_{j=1}^{K_i-1} \hat{d}(s_{t_{i,j}}, s_{t_{i,j+1}})
\end{equation}
and $\hat{d}(s_i, s_j) = d_1(s_i, s_j)$ is the initial heuristic of distance between $s_i$ and $s_j$. This fitness function is in similar form as $(J_2 + J_3)^{-1}$ as it aims to reduce the average and bottleneck distance cost simultaneously. If $(J_2 + J_3)^{-1}$ is at global maximum, then $J_2, J_3$ are both at their global minimum.
\subsubsection{Crossover and Mutation} The crossover operation is used to reproduce offsprings $C_{c1}$ and $C_{c2}$ from two parents $C_{p1}$ and $C_{p2}$. \autoref{fig:crossover} shows an example of this operation. Two crossover points $i, j$ are randomly chosen. Then the $i$-th to $j$-th genomes of $C_{p1}$, labeled in blue, are preserved to the child, and the unused genomes $C_{p2}$, labeled in red, will fill the remaining genomes of $C_{c1}$. The second offspring $C_{c2}$ is reproduced with the same procedure with reversed roles of $C_{p1}$ and $C_{p2}$. 
\begin{figure}[hbt]
	\centering
	\begin{tabular}{l|cccccccccc}
	Parent 1 & 3 & -2 & \textcolor{blue}{1} & \textcolor{blue}{2} & \textcolor{blue}{5} & \textcolor{blue}{6} & 4 & -1 & 7 & -3	\\
	Parent 2 & \st{6} & \st{2} & \textcolor{red}{-1} & \textcolor{red}{4} & \textcolor{red}{3} & \textcolor{red}{-3} & \textcolor{red}{7} & \textcolor{red}{-2} & \st{5} & \st{1}	\\
	Child 1 & \textcolor{red}{-1} & \textcolor{red}{4} & \textcolor{blue}{1} & \textcolor{blue}{2} & \textcolor{blue}{5} & \textcolor{blue}{6} & \textcolor{red}{3} & \textcolor{red}{-3} & \textcolor{red}{7} & \textcolor{red}{-2}
	\end{tabular}
	\caption{Example of chromosome crossover operation}
	\label{fig:crossover}
\end{figure} \\
The mutation operation is relatively simpler. Two random points $m,n$ are first chosen. Then one randomly permute the $m$-th to $n$-th genomes. This is also called the scramble mutation.
\subsubsection{Learning Rule}
After the robots have finished the paths from any $s_i$ to $s_j$, the actual travelled distance $D_i(s_i, s_j)$ will be used to update the heuristic $\hat{d}(s_i, s_j)$. To generalize the learning rule, let $\mathbf{\hat{D}} = [\hat{d}_{ij}]$ where $\hat{d}_{ij} = \hat{d}(s_i, s_j)$ and $\mathbf{A} = [a_{ij}]$ where $a_{ij} = \lambda D_i(s_i, s_j) + (1-\lambda) \hat{d}_{ij}$. The $\lambda$ is a binary indicator of whether or not an update is received. We consider a quadratic error cost function $J$, given by
\begin{equation}
	J = \dfrac{1}{2} \text{Tr}(\mathbf{A} - \mathbf{\hat{D}})'(\mathbf{A} - \mathbf{\hat{D}})
\end{equation}
and apply gradient descent to update the heuristics:
\begin{equation}
	\Delta \hat{d}_{ij} = - \eta \dfrac{\partial J}{\partial \hat{d}_{ij}} =  \eta (a_{ij} - \hat{d}_{ij})
\end{equation}
\begin{equation}
	\hat{d}_{ij} \leftarrow  \hat{d}_{ij} + \Delta \hat{d}_{ij}
\end{equation}
where $\eta$ is the learning rate.
%
%
%
%
\subsection{Recursive Excitation/Relaxation APF Algorithm}
In this section, we will present a novel recursive excitation/relaxation APF algorithm that completely eliminates the well-known problem of local minima with proof, which will be discussed in more details.
\subsubsection{Conventional Artificial Potential Field Algorithm}
 The conventional APF algorithm is based on the potential function in the form of \eqref{eq:conventional_APF} \cite{mrpp_apf}.
\begin{equation}
	\label{eq:conventional_APF}
	U_i(k,s) =  U_i^{att}(s) + U_i^{rep}(s)	
\end{equation} where $U_i(k,s)$ is the potential at position $s$ and time $k$ for robot $r_i$, $U_i^{att}(s)$ and $U_i^{rep}(s)$ are the attractive and repulsive potentials. Typically, the attractive potential is contributed by the goal position while the repulsive potential by obstacles and other robots to avoid collision \cite{mrpp_apf_animal}. Also, the potential from other robots is a function of time as robots are dynamic objects. Equation \eqref{eq:conventional_APF} can therefore be rewritten as \eqref{eq:conventional_APF2}.
\begin{equation}
	\label{eq:conventional_APF2}
	U_i(k,s) = U_i^g(s) + U_i^o(s) + U_i^r(k,s)
\end{equation}
where $U_i^g(s), U_i^o(s), U_i^r(k,s)$ are the potential components from goal, obstacle and other robots. The potential component at a point $s$ under influence of $s'$, denoted by $\phi(s, s')$, regardless of types, is generally expressed as a function of $d_p(s, s')$ \cite{mrpp_apf_guideline}.
\begin{equation}
	\label{eq:general_potential}
	\phi(s, s')	=	g\big(d_p(s,s')\big)
\end{equation}
where $g(x)$ is monotonically increasing for attraction and non-increasing for repulsion.
At each $k$, the robot selects a neighbor $s' \in \mathcal{N}(s_{r_i}(k))$ as the next step based on discrete gradient descent \cite{rpp_apf_num}, a modified version is shown in \eqref{eq:robot_pos_update}.
\begin{equation}
	\label{eq:robot_pos_update}
	s_{r_i}(k+1)	=	\argmin_{s' \in \mathcal{N}(s_{r_i}(k))} U_i(k, s')
\end{equation}
\subsubsection{Proposed Artificial Potential Field Algorithm}

While the conventional APF algorithm is highly efficient for path finding in various applications, it has a fatal limitation, the existence of local minima \cite{mrpp_apf}.
\begin{definition}
	\label{def:local_min}
	A position $s_m$ is a local minimum for robot $r_i$ at time $k$ iff $\forall s' \in \mathcal{N}_a(s_m)$, $U_i(k, s_m) \leq U_i(k, s') $.
\end{definition}
When a robot is at a local minimum $s_m$, it will be trapped indefinitely according to \eqref{eq:robot_pos_update}. This is considered a key disadvantage of the APF algorithm \cite{mrpp_apf_regression}.
The proposed RERAPF algorithm is an extension of the conventional APF algorithm designed to overcome the local minima problem. In this algorithm, two operations are introduced, namely the Excitation $f_e$ and Relaxation $f_r$ defined in \eqref{eq:excitation} and \eqref{eq:relaxation}.
\begin{eqnarray}
	\label{eq:excitation}
	f_e(x(k)) &=& \gamma x(k-1) \\
	\label{eq:relaxation}
	f_r(x(k)) &=& (1-\alpha)x(k-1) + \alpha x(k_0)
\end{eqnarray} where $k_0$ is an arbitrary starting time, also $\gamma \in (1,\infty)$ and $\alpha \in [0,1)$ are called the excitation factor and relaxation factor. Rewriting \eqref{eq:conventional_APF} to \eqref{eq:APF2}:
\begin{equation}
	\label{eq:APF2}
	U_i(k,s) = U_i^s(s) + U_i^d(k,s)
\end{equation} where $U_i^s(s) = U_i^g(s) + U_i^o(s)$ is the potential contributed by static objects and $U_i^d(k, s) = U_i^r(k, s)$ by dynamic objects, hereinafter referred to as the static potential and dynamic potential respectively. 
Instead of being a time-independent function, the static potential is modified into a time-recursive function. The potential update is shown in \eqref{eq:potential_update}.
\begin{equation}
	\label{eq:potential_update}
	U_i^s(k,s) = 
	\begin{cases}
		f_e(U_i^s(k-1,s)), & s_{r_i}(k) = s								\\
		f_r(U_i^s(k-1,s)),	&	otherwise
	\end{cases}	
\end{equation} $\forall s \in \mathcal{N}(s_{r_i}(k))$ with initial condition $U_i^s(k_0,s) = U_i^s(s)$. The purpose of excitation is to increase the potential at the position of the robot $r_i$, i.e., $U_i(k, s_{r_i}(k))$, hence attempting to remove any local minimum in case the robot $r_i$ is trapped. For each position $s$, we only calculate the static potential once until the current task is finished. The explored positions are stored in a set $\mathcal{E}_i$ for each robot $r_i$. We now formally present the RERAPF algorithm in Algorithm \ref{algo:rapf}. 

\setlength{\textfloatsep}{0pt}
\begin{figure}[h]
 \removelatexerror
  \begin{algorithm}[H]
  	\label{algo:rapf}
	\caption{Proposed RERAPF Algorithm}
	$k \leftarrow 0$ \\
	\For {$i \in [1,N]$} {
		$\mathcal{E}_i \leftarrow \emptyset$ 	\\
	}
	\While {not all robots reached goals} {
		\For{each $r_i \in \mathcal{R}$ with $T_i \neq \emptyset$} {
			\uIf {$s_{r_i}(k) \neq s_g^i$} {
		      \For {$s' \in \mathcal{N}(s_{r_i}(k))$} {
  				\uIf {$s' \not \in \mathcal{E}_i$} {
  					$U_i^s(k, s') \leftarrow U_i^g(s') + U_i^o(s')$	\\
  					$\mathcal{E}_i \leftarrow \mathcal{E}_i \cup \{s'\}$ \\
				}
				\uElseIf{$s' = s_{r_i}(k)$}{
		      		$U_i^s(k,s') \leftarrow f_e(U_i^s(k-1,s))$	\\
		      	}
		      	\Else{
		      		$U_i^s(k,s') \leftarrow f_r(U_i^s(k-1,s))$	\\
		      	}
	      		$U_i(k, s') \leftarrow U_i^s(k, s') + U_i^d(k, s')$ \\
		      }
		      $s_{r_i}(k+1) \leftarrow \argmin_{s' \in \mathcal{N}(s_{r_i}(k))} U_i(k,s')$	\\
			}
			\Else{
				$T_i \leftarrow T_i \backslash \{t_{i,1}\}$ \\
			}
	   }
	   $k \leftarrow k+1$ \\
	}
  \end{algorithm}
\end{figure}
\begin{definition}
	A local minima $s_m$ is temporary to $r_i$ if $\forall k_0$ $\exists k' \in (k_0, \infty)$ $\exists s' \in \mathcal{N}_a(s_m)$ s.t. $U_i(k',s') < U_i(k',s_m)$.
\end{definition}
\begin{assumption}
	\label{assum:apf_pos}
	The static potential function satisfies $U_i^s(k,s) \geq 0$ $\forall i,k,s$ with unique global minimum at $s=s_g^i$.
\end{assumption}
\begin{assumption}
	\label{assum:finite_dynamic_potential}
	The dynamic potential function satisfies $U_i^d(k,s) \in  [0,\infty)$ $\forall i,k,s$, i.e., always non-negative and finite, and $U_i^d(k,s) \gg U_i^s(s')$ if $\exists r \neq r_i$ $s_r(k) = s$ $\forall i,k,s,s'$.
\end{assumption}
Before we introduce the main theorem in this paper, we first introduce the following lemma about local minimum:
\begin{lemma}
	\label{lemma:local_min_temp}
	Any local minimum reached by robot $r_i$ is temporary to $r_i$ in the proposed RERAPF algorithm.
\end{lemma}
\begin{proof}
	Under assumption \ref{assum:apf_pos} and \ref{assum:finite_dynamic_potential}, consider an arbitrary local minimum $s_m$ in which robot $r_i$ is trapped at time $k_t$, i.e., $s_{r_i}(k_t)=s_m$. Also assume there exists an unoccupied neighbor for $s_m$. Then by the definition \ref{def:local_min}, $\forall s' \in \mathcal{N}_a(s_m)$, $U_i(k_t,s_m) \leq U_i(k_t,s')$ and $s_{r_i}(k_t+1)=s_m$. It is also possible that $s_{r_i}(k_t+1) = s'$ if $U_i(k_t,s')=U_i(k_t,s_m)$ which immediately implies that the robot has escaped from the local minimum $s_m$. Therefore we will focus on the case where $s_{r_i}(k_t+1)=s_m$.
	\begin{eqnarray}
		U_i(k_t+1,s_m) &=& f_e(U_i(k_t,s_m)) + U_i^d(k_t+1,s_m) \quad	\\
		U_i(k_t+1,s') &=& f_r(U_i(k_t,s')) + U_i^d(k_t+1,s')
	\end{eqnarray}
	If $\forall s' \in \mathcal{N}_a(s_m)$, $U_i(k_t+1,s_m)\leq U_i(k_t+1,s')$ then the above update rules will be repeated until $\exists k'>k_t, s' \in \mathcal{N}_a(s_m)$, $U_i(k',s_m)>U_i(k',s')$. The dynamic potential $U_i^d(k,s)$ is bounded with $\sup U_i^d(k,s)=m$ and $\inf U_i^d(k,s)=n$ where $m$ and $n$ are two non-negative finite numbers  that $m - n \geq 0$. Consider the extreme worst case scenario where $U_i^d(k,s_m)=n$ and $U_i^d(k,s')=m$ $\forall s' \in \mathcal{N}_a(s_m)$. For some $k'>k_t$,
	\begin{eqnarray}
		U_i(k',s_m)&=&\gamma^{k'-k_t}U_i^s(k_t,s_m) + n	\\
		U_i(k',s')&=&(1-\alpha)^{k'-k_t}U_i^s(k_t,s')	\nonumber \\
					&&{}+[1-(1-\alpha)^{k'-k_t}]U_i^s(s')+m
	\end{eqnarray}
	Let $s^* = \argmin_{s' \in \mathcal{N}_a(s_m)} U_i(k',s')$. If $r_i$ is to escape from $s_m$, the inequality \eqref{eq:k_ineq_unsol} must be satisfied.
	\begin{eqnarray}
		\label{eq:k_ineq_unsol}
		U_i(k',s_m) &>& U_i(k',s^*)	\\
		\label{eq:k_ineq_unsol2}
		\gamma^{\Delta k}U_i^s(k_t,s_m) &>& (1-\alpha)^{\Delta k} (U_i^s(k_t,s^*)-U_i^s(s^*))	\nonumber \\ 
		&& {}+ U_i^s(s^*) + m - n
	\end{eqnarray}
	where $\Delta k = k'- k_t$. It can be easily verified that $U_i^s(k_t, s^*) \geq (1-\alpha)^{\Delta k} (U_i^s(k_t,s^*) - U_i^s(s^*)) + U_i^s(s^*)$.
%
	Then consider the following inequality,
	\begin{equation}
		\label{eq:k_ineq}
		\gamma^{\Delta k'} U_i^s(k_t,s_m) > U_i^s(k_t,s^*) + m - n
	\end{equation}
	where $\Delta k' \geq \Delta k$. The inequality \eqref{eq:k_ineq} has a simple analytical solution for the critical point of $\Delta k'$ where $\Delta k' > (\Delta k')_{crit}$ is the region of solutions.
	\begin{equation}
		\label{eq:k_over}
		(\Delta k')_{crit} = \log_\gamma \dfrac{U_i^s(k_t,s^*) + m - n}{U_i^s(k_t,s_m)} \geq (\Delta k)_{crit}
	\end{equation}
Similarly, $\Delta k > (\Delta k)_{crit}$ is the region of solutions	 for \eqref{eq:k_ineq_unsol2}. Note that is an overestimation of the solution of \eqref{eq:k_ineq_unsol} for general case (where the dynamic potentials are not at extreme values), thus an upper bound. In other words, the time required to eliminate local minimum $s_m$ does not exceed $(\Delta k')_{crit} + 1$ regardlessly. Since $\exists \Delta k \leq (\Delta k')_{crit}+1 < \infty$ such that $U_i(k',s_m)>U_i(k',s^*)$, the local minimum $s_m$ is temporary to $r_i$.
	\end{proof}
	Lemma \ref{lemma:local_min_temp} shows that the introduction of excitation can remove local minimum in case a robot is trapped. From (25) of \cite{Tsang2018}, it is observed that the upper bound of $\Delta k$ decreases with increasing $\gamma$, meaning that a larger excitation factor shortens entrapment time. In addition, the introduction of relaxation factor also facilitates the escape of local minimum as it reduces the potentials of the neighbors to the original values over time.
\subsubsection{Semi-Completeness for Proposed RERAPF Algorithm}
The proposed RERAPF algorithm is capable of eliminating local minimum as proven in Lemma \ref{lemma:local_min_temp}, which results in semi-completeness shown in the following.
\begin{definition}
	An algorithm is semi-complete if and only if it is guaranteed to find a solution when there exists one but may not return false when there is none.
\end{definition}
\begin{theorem}
	\label{thm:main}
	For every $\alpha \in [0,1)$, $\exists \gamma \in (1,\infty)$ such that the proposed RERAPF algorithm is semi-complete.
\end{theorem}
\begin{proof}
	Assume robot $r_i$ is contained in a fixed region $V \subset \mathcal{S}$ indefinitely where $s_g^i \not \in V$ and $s_{r_i}(k) \in V$ $\forall k$. Also, we apply an additional constraint on $V$ such that all positions in $V$ must be reached by $r_i$ at some $k$, i.e., $\forall s \in V$ $\exists k$  $s_{r_i}(k)=s$ with finite maximum time intervals. Lemma \ref{lemma:local_min_temp} implies that the additional constraint is always achievable for small region $V$ and $\vert V \vert > 1$. For each $s \in V$, the corresponding potential can be modeled as a recursive Bernoulli random process with $P[s_{r_i}(k)=s]=p$. Because of the constraint on $V$, $p$ is not the true probability but rather a relative frequency, due to the assumption that all states will be reached by $r_i$ with finite interval, and therefore $p \in (0,1)$. 
	\begin{eqnarray}
		\mathbb{E}[U_i^s(k,s)]	&=&	p\mathbb{E}[f_e(U_i^s(k-1,s))] \nonumber \\
						& & +{}(1-p)\mathbb{E}[f_r(U_i^s(k-1,s))]		\\
						&=&	\bigg[\beta^{k-k_0} + (1-p)\alpha \sum_{i=0}^{k-k_0-1} \beta^i\bigg]U_i^s(s)
	\end{eqnarray}
	where $\beta = p\gamma+(1-p)(1-\alpha)$. It can be verified that $\lim_{k\rightarrow \infty} \mathbb{E}[U_i^s(k,s)] \rightarrow \infty$ when $\beta > 1$, and therefore $\gamma > 1 - \alpha + \alpha/p$. This implies that $U_i^s(k,s)$ will tend to increase, albeit not necessarily monotone, with large $k$. Since $p \in (0,1)$, there must exist a finite $\gamma$ satisfying this inequity.
	Assume the condition $\gamma > 1 - \alpha + \alpha/p$ is satisfied. If there exists a solution, then $\exists k'' \in (k_0, \infty)$ such that $\exists s'' \in \mathcal{N}_a(V)$ with finite non-increasing potential. For some finite $k^* > k''$ when $s_{r_i}(k^*) \in \mathcal{N}_a(s'') \cap V$, then $s'' \in \mathcal{N}_a(s_{r_i}(k^*))$ with $U_i(k^*, s'') < U_i(k^*, s_{r_i}(k^*))$, which implies that $s_{r_i}(k^*+1) \not \in V$ as at least one neighbor of $s_{r_i}(k^*)$ has lower potential.
	By contradiction, it is impossible for $r_i$ to be contained in a fixed region $V$ indefinitely. In other words, if $r_i$ is contained in a region $V$ indefinitely, $V$ must be expanding until $s_g^i \in V$. By definition of $V$, the robot $r_i$ will reach $s_g^i$ in finite time. 
	However, if there is no possible path to the goal, it will continue move inside $V$ indefinitely without reaching $s_g^i$, thus only semi-complete.
\end{proof}
\begin{figure*}[!htb]
\centering
	\begin{subfigure}[b]{0.3292\textwidth}
	\includegraphics[width=\textwidth]{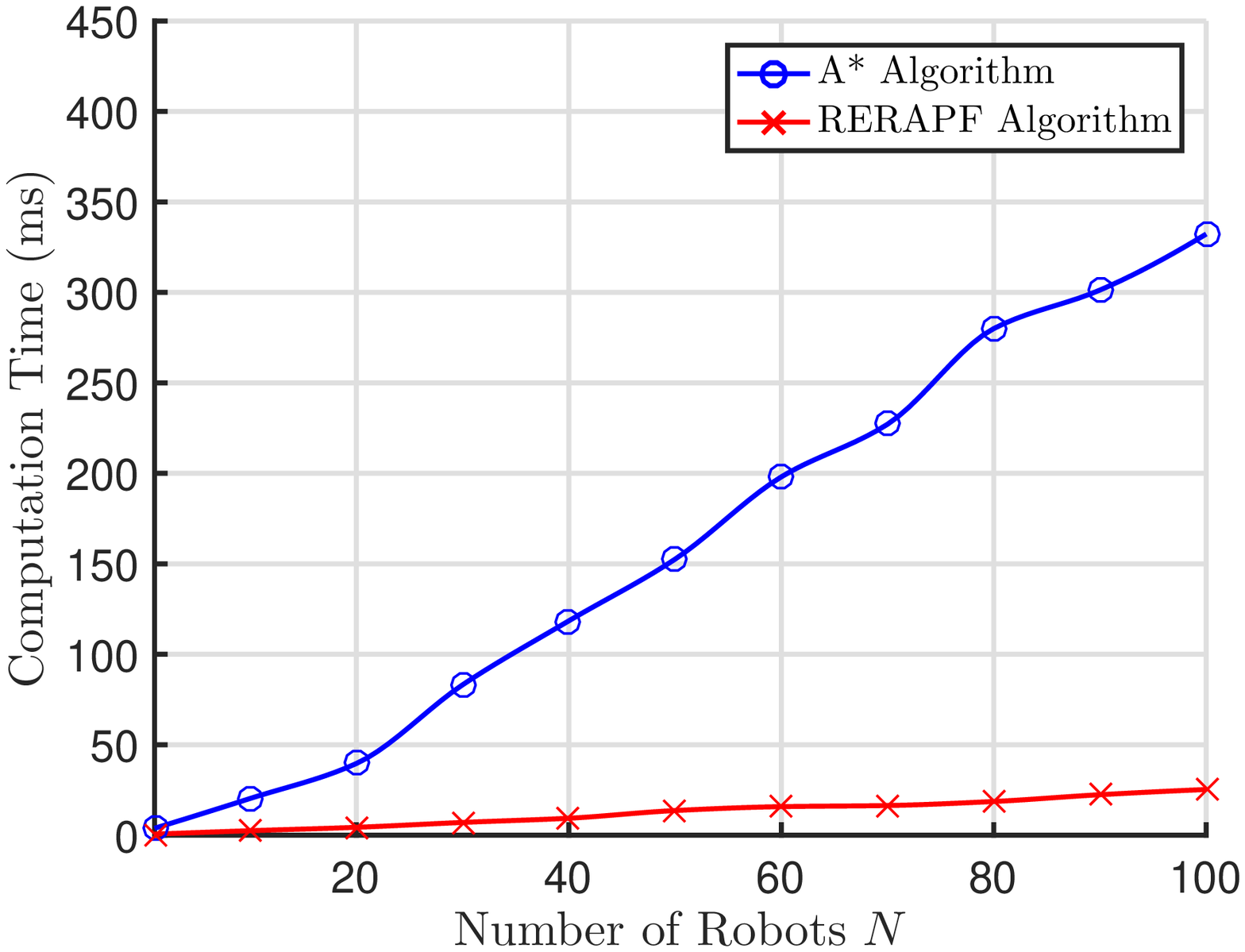}
	\caption{Computational time}
	\label{fig:comp_time}
	\end{subfigure}
	\begin{subfigure}[b]{0.3292\textwidth}
	\includegraphics[width=\textwidth]{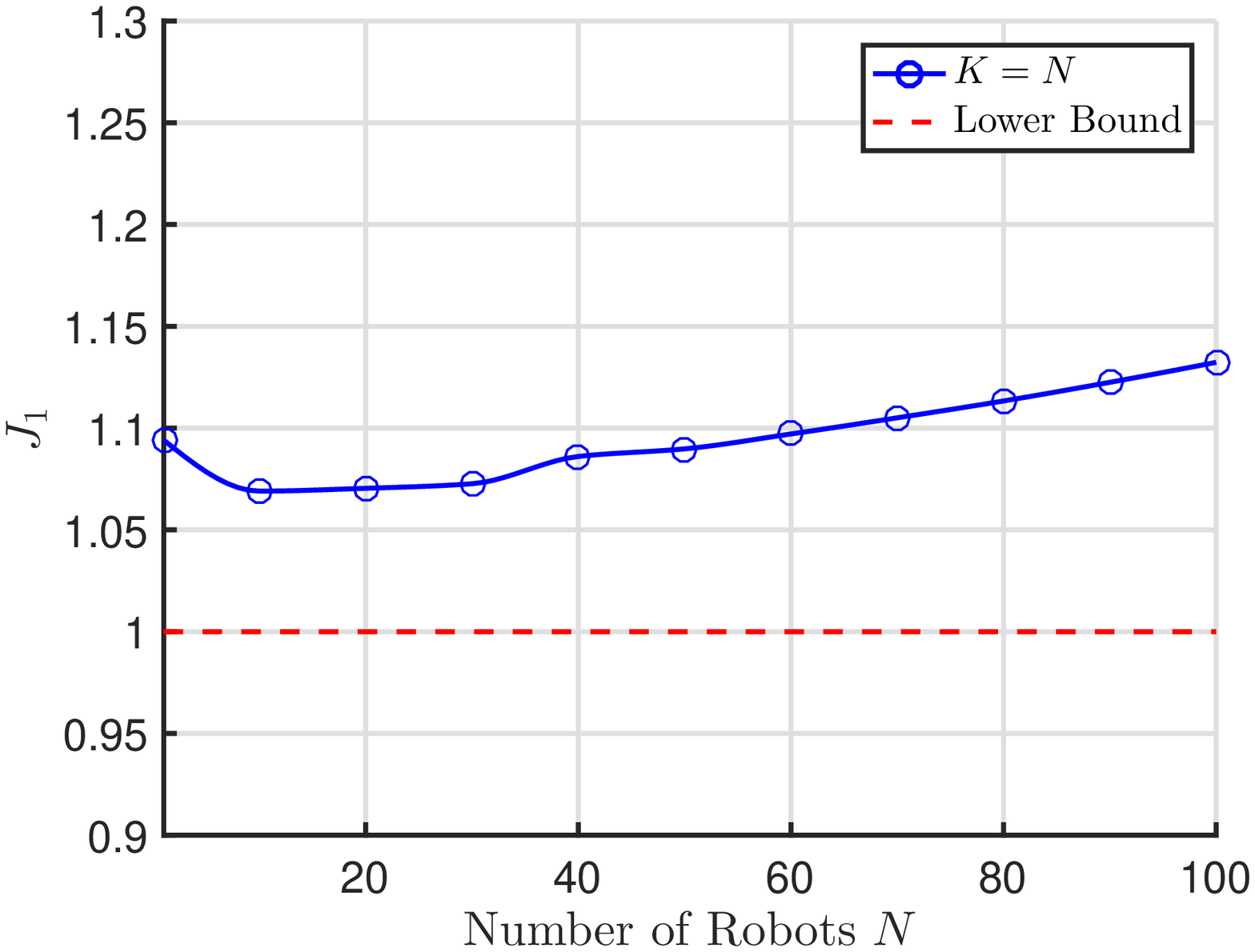}
	\caption{$J_1$ with $K=N$}
	\label{fig:f1}
	\end{subfigure}
	\begin{subfigure}[b]{0.3292\textwidth}
	\includegraphics[width=\textwidth]{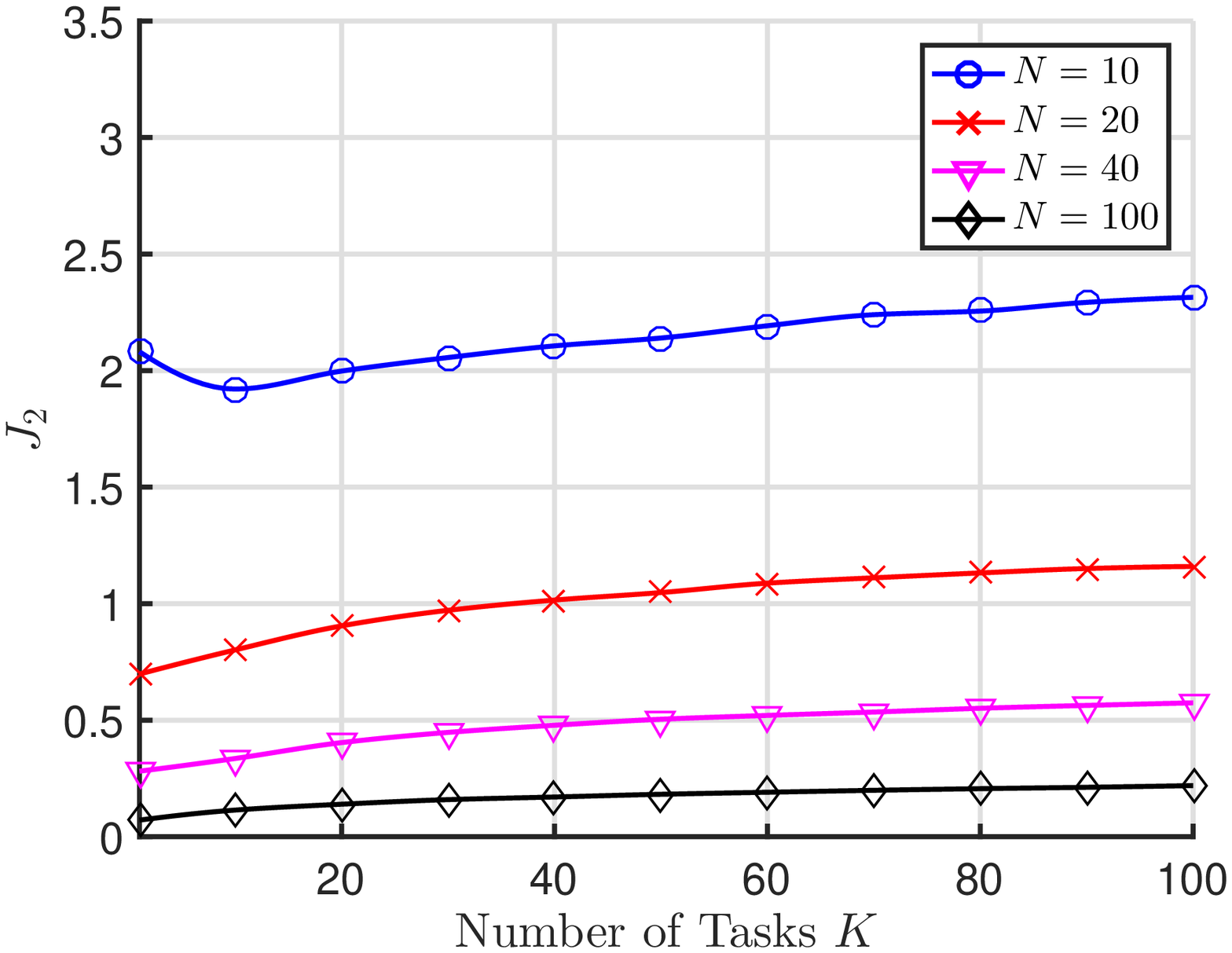}
	\caption{$J_2$ with $N=10,20,40,100$}
	\label{fig:f2}
	\end{subfigure}
	\begin{subfigure}[b]{0.3292\textwidth}
	\includegraphics[width=\textwidth]{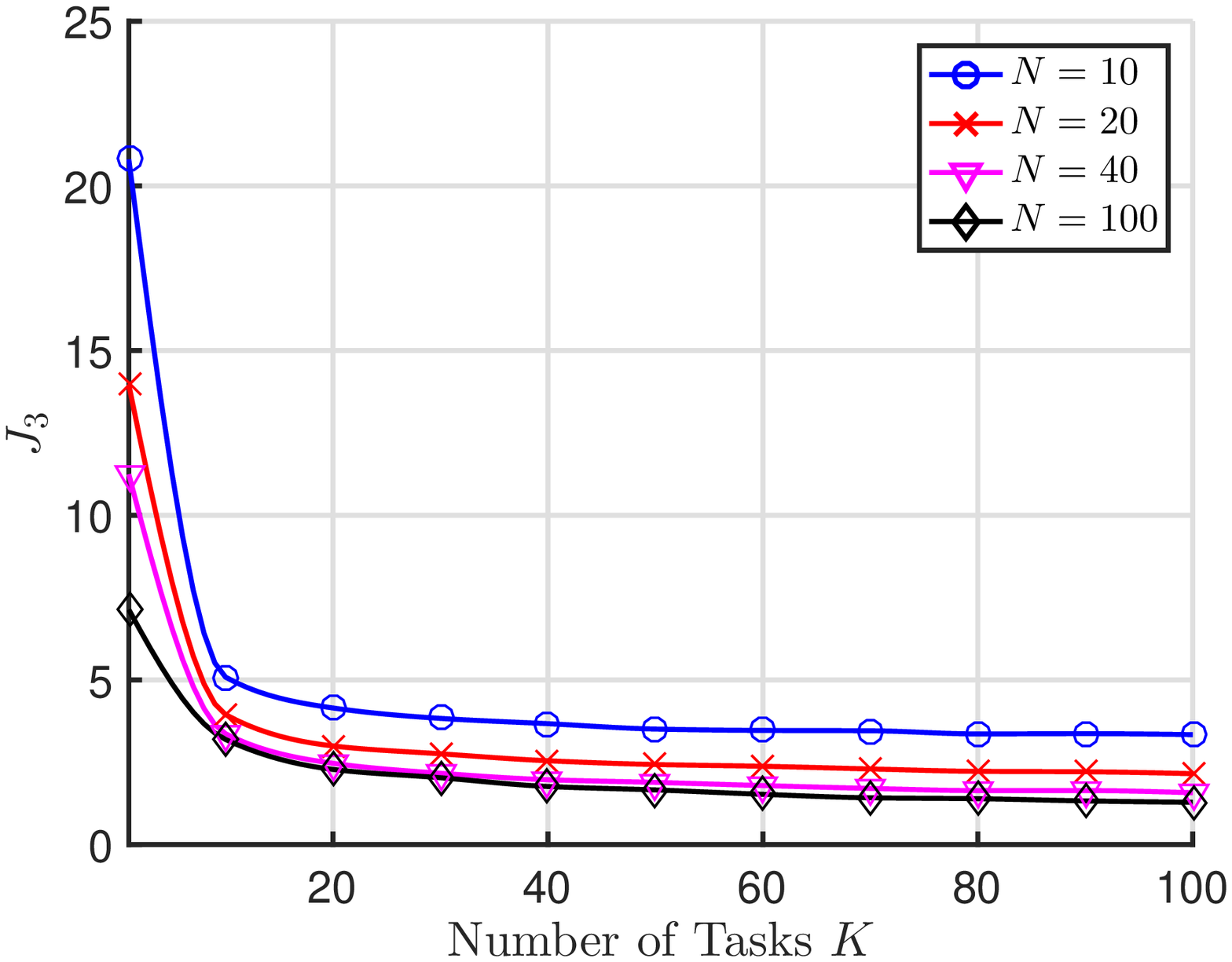}
	\caption{$J_3$ with $N=10,20,40,100$}
	\label{fig:f3}
	\end{subfigure}
	\begin{subfigure}[b]{0.3292\textwidth}
	\includegraphics[width=\textwidth]{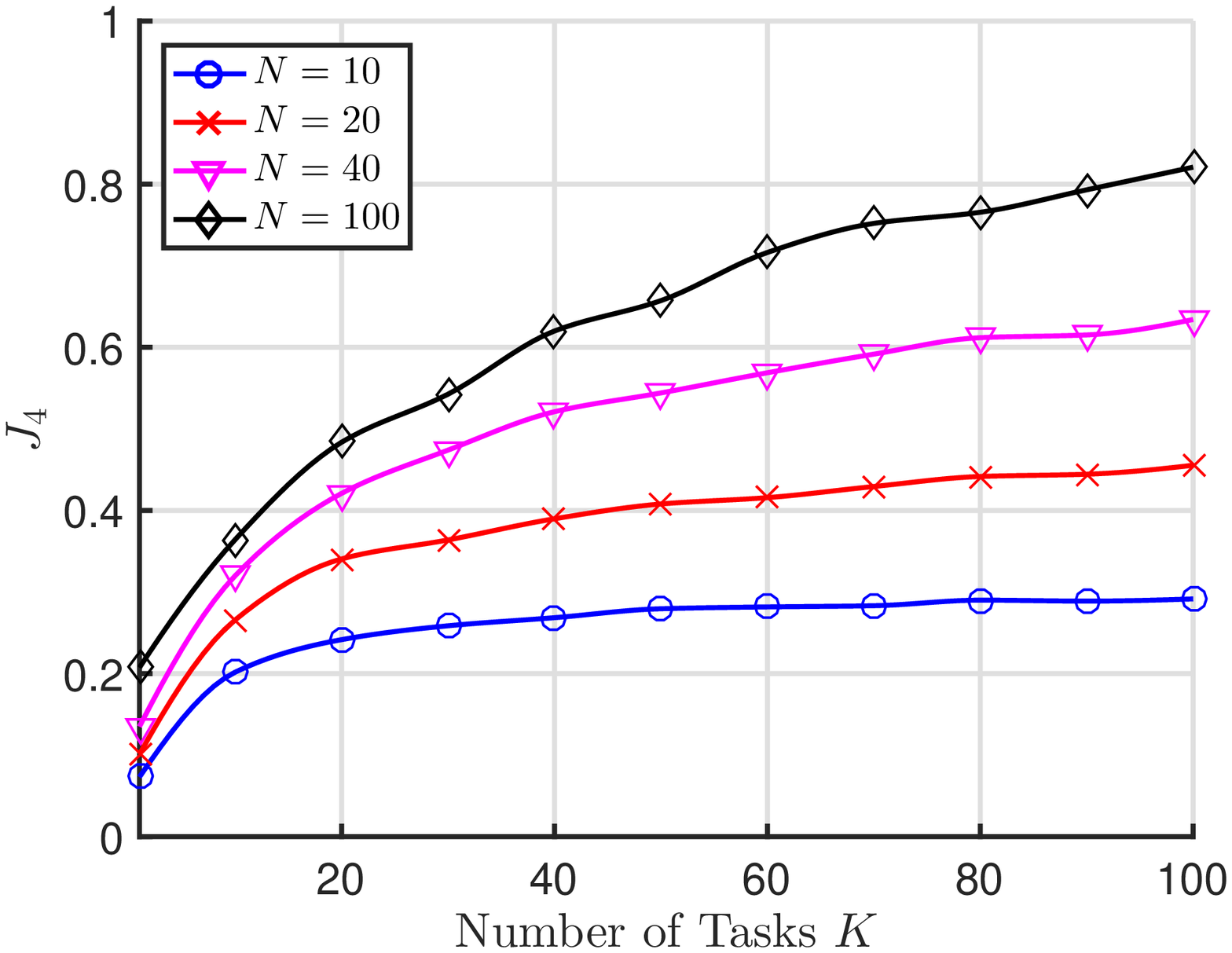}
	\caption{$J_4$ with $N=10,20,40,100$}
	\label{fig:f4}
	\end{subfigure}
  \caption{Simulation Results}
  \label{fig:sim_results}
\end{figure*}
\noindent
\begin{remark}
	Although a closed-form inequality was derived in \autoref{thm:main}, it is impractical to design $\alpha$ and $\gamma$ directly based on the condition because $p$ is unknown. The parameters still need to be chosen by trial and error. However, it is clear that $p$ decreases as $\vert V \vert$ increases. In other words, a larger layout should lead to smaller $p$, hence larger $\gamma$ and/or smaller $\alpha$ required. 
\end{remark}

\subsubsection{Implementation} The main implementation challenge is the design of potential functions. In the following context, we introduce an implementation based on proximity sensors information. We firstly rewrite the general potential function \eqref{eq:general_potential} as the weighted sum of the $p^j$-norms, addition to a constant $\varepsilon^j$, to the power of $e^j$. The superscript $j$ denotes the source of the potential ($g$ for goals, $o$ for obstacles, $r$ for robots), as different object may produce a different potential function.
\begin{equation}
	\phi^j(s,s')	=	\sum_{i=0}^\infty c_i^j \Big(d_{p_i^j}(s,s') + \varepsilon_i^j\Big)^{e_i^j}
\end{equation}
where $p_i^j, c_i^j, e_i^j, \varepsilon_i^j$ are design parameters. More specifically, $\varepsilon_i^r > 0$ $\forall i \in \{ x \colon c_x^r \neq 0 \}$ in order to bound the dynamic potential function to be finite.
To utilize the local sensors information, we introduce an alternative form of the potential function as follows. 
\begin{equation}
	\varphi^j(s, s')	=
	\begin{cases}
		\phi^j(s,s'),	& s' \in \mathcal{Q}(s)	\\
		0,	&	otherwise
	\end{cases}	
\end{equation}
where $\mathcal{Q}(s) = \bigcap_{s^* \in \mathcal{N}(s)} \mathcal{P}(s^*)$ to ensure a consistent static potential calculated at different location. This can prevent duplicate computation of potential at the same position, and ensure that the excitation/relaxation operations are correctly performed. The potential function is highly dependent on the sensor information for robot $r_i$ at time $k$. 
\begin{eqnarray}
	U_i^s(s)	&=&	\phi^g(s, s_g^i) + \sum_{o \in \mathcal{O}} \varphi^o(s, s^o)	\\
	U_i^d(k, s) &=& \sum_{r \in \mathcal{R}\backslash\{r_i\}} \varphi^r(s, s_r(k))
\end{eqnarray}

It is clear that this implementation can satisfy all assumptions made previously on the potential functions with proper $p_i^j, c_i^j, e_i^j, \varepsilon_i^j$.
The advantage of this approach is the adaptivity to change of the warehouse environment, for example, change of layout or unexpected obstacles.

\section{Simulation Results}
\label{section:sim}

In this section, we present some simulation results for the proposed warehouse automation system based on the cost functions introduced in \autoref{section:problem}. The simulation parameters are listed in \autoref{table:sim_par}, whilst the layout used in the simulation is in the same pattern as \autoref{fig:layout}. For each set of parameters, 500 runs of simulation were conducted to obtain the corresponding results.
\begin{table}[h]
	\centering
	\renewcommand{\arraystretch}{1.2}
	\begin{tabular}{|c|c|}
		\hline
		Layout Size	&	$81\times 80$	\\	\hline
		Number of Robots		&	$1 - 100$	\\ \hline
		Number of Tasks		&	$1 - 100$	\\ \hline
		Excitation Factor $\gamma$	&	$15$		\\	\hline
		Relaxation Factor $\alpha$	&	$0.05$		\\	\hline
	\end{tabular}
	\caption{Simulation parameters}
	\label{table:sim_par}
\end{table}

For the sake of simplicity, the following potential functions are adopted in the simulation.
\begin{align*}
	U_i^s(s)	&=	\phi^g(s, s_g^i) + \sum_{o \in \mathcal{O}} \varphi^o(s, s^o)	\\
	U_i^d(k, s) &= 0.01\sum_{r \in \mathcal{R}\backslash\{r_i\}}  \varphi^o(s, s_r(k))	\\
	\phi^g(s, s_g^i)	&=	d_\infty(s,s_g^i)	\\
	\varphi^o(s, s^o)	&=	
	\begin{cases}
		0.1 (d_2(s, s^o) + \varepsilon)^{-2} ,	& s^o \in \mathcal{Q}(s)	\\
		0,	&	otherwise
	\end{cases}		
\end{align*}
where $\varepsilon = 10^{-9}$ is an arbitrarily chosen small number, and $\mathcal{Q}(s) = \bigcap_{s^* \in \mathcal{N}(s)} \mathcal{P}(s^*)$.

In addition to the cost functions mentioned previously, we also evaluate the computational time of the proposed RERAPF algorithm compared with traditional A* algorithm. The simulation results are shown in \autoref{fig:sim_results}. It is shown that the computation time of A* algorithm is approximately $13$ times higher than the RERAPF algorithm on average. Also, since the proposed algorithm can be executed distributively on each robot, the actual computational load could be much lower for each robot. As for the path cost $J_1$, it is below $1.2$ within the simulation environment, i.e., the difference between the actual path length and the optimal length is less than $20\%$, with a maximum slope of $0.09\%$ per additional robot. This is the tradeoff between the computation time and optimality for the path planning algorithm.

The results for the cost functions $J_2$, $J_3$ and task completion rate $J_4$ with $N=10,20,40,100$ are shown in \autoref{fig:f2} to \autoref{fig:f4}. It is seen that all of $J_2$, $J_3$ and $J_4$ gradually converges as $K$ increases albeit with a slower converging rate for larger $K$, with diminishing marginal improvement with respect to increase in $N$. It is interesting to notice that $J_3$ and $J_4$ can be improved with larger $K$, and $J_2$ was slightly worsen as $K$ increases. The intuition behind this is that when the number of tasks is sufficiently large, it is more likely for the workload of robots distributed more uniformly, hence a lower bottleneck cost $J_3$. Also, multiple tasks can be completed by a single robot at once, hence the tasks done more efficiently, which leads to a lower $J_3$ and higher task completion rate $J_4$. On the other hand, when there are significantly more robots than tasks, the probability that there exist a robot close to the tasks is relatively high, therefore resulting in a lower average travel distance $J_2$.

\section{Conclusion and Future Work}
In this paper, we have presented a computationally efficient local path planning algorithm, namely the RERAPF algorithm, and a genetic task allocation algorithm with heuristic learning rule. The semi-completeness of the RERAPF algorithm was proven by showing that it is impossible for any robot to be trapped in any fixed region that does not contain the goal position. The performances of each subsystem, and the overall system were shown in simulation. Future work may include consideration of additional practical constraints, such as battery and robot capacities. Also, an equivalent RERAPF algorithm of the continuous time or continuous space configuration may also be considered to derive a more general algorithm.

\addtolength{\textheight}{0cm}  



%
%

\bibliographystyle{IEEEtran}
\bibliography{IEEEabrv,bibliography}

\end{document}